\documentclass[a4paper,UKenglish,cleveref, autoref, thm-restate]{lipics-v2021}

\pdfoutput=1 
\hideLIPIcs  

\usepackage{amsmath,amssymb}
\usepackage{hyperref}

\usepackage{xspace}
\usepackage{footmisc}
\usepackage{tikz}
\usetikzlibrary{arrows.meta,shapes}
\newcommand{\op}{\operatorname}

\newcommand{\arity}{\operatorname{ar}}

\newcommand{\FO}{\ensuremath{\operatorname{FO}}\xspace}
\newcommand{\MSO}{\ensuremath{\operatorname{MSO}}\xspace}
\newcommand{\CMSO}[1][]{\ensuremath{\operatorname{CMSO}}\xspace}

\newcommand{\lvlParity}{\delta}

\newcommand{\colour}{\mathsf{col}}

\newcommand{\repSet}{\mathsf{REP}}
\newcommand{\rep}{\mathsf{rep}}
\newcommand{\leader}{\mathsf{leader}}
\newcommand{\single}{\mathsf{single}}
\newcommand{\fully}{\mathsf{fully}}
\newcommand{\missed}{\mathsf{missed}}
\newcommand{\parent}{\mathsf{PARENT}}
\newcommand{\child}{\mathsf{CHILD}}
\newcommand{\ancestor}{\mathsf{ANC}}
\newcommand{\descendant}{\mathsf{DESC}}

\newcommand{\childT}{\mathsf{child}}
\newcommand{\ancestorT}{\mathsf{anc}}
\newcommand{\descendantT}{\mathsf{desc}}

\newcommand{\rootT}{\mathsf{root}}
\newcommand{\rootNode}{\mathsf{ROOT}}

\newcommand{\leafT}{\mathsf{leaf}}
\newcommand{\leaf}{\mathsf{LEAF}}
\newcommand{\SET}{\mathsf{SET}}
\newcommand{\BIPART}{\mathsf{BIPART}}
\newcommand{\copyV}{\mathsf{copy}}
\newcommand{\colorV}{\mathsf{colour}}
\newcommand{\realization}{\mathsf{real}}
\newcommand{\cT}{\mathcal{T}}
\newcommand{\cL}{\mathcal{L}}

\newcommand{\LHS}{\mathsf{Left}}
\newcommand{\RHS}{\mathsf{Right}}
\newcommand{\LL}{\mathsf{LL}}
\newcommand{\LR}{\mathsf{LR}}
\newcommand{\RL}{\mathsf{RL}}
\newcommand{\RR}{\mathsf{RR}}

\newcommand\DEGENERATE{\ensuremath{\mathsf{degenerate}}\xspace}
\newcommand\PRIME{\ensuremath{\mathsf{prime}}\xspace}
\newcommand\LINEAR{\ensuremath{\mathsf{linear}}\xspace}

\title{The role of counting quantifiers in laminar set systems} 

\author{Rutger Campbell}{University of Waterloo, Waterloo, Canada}{rtrcampb@uwaterloo.ca}{}{}

\author{Noleen K\"{o}hler}{University of Leeds, Leeds, UK}{N.Koehler@leeds.ac.uk}{https://orcid.org/0000-0002-1023-6530}{}

\authorrunning{R. Campbell and N. K\"{o}hler} 

\Copyright{Jane Open Access and Joan R. Public} 



\category{} 

\relatedversion{} 

\acknowledgements{Campbell was supported in part by the National Research Foundation of Korea (NRF) grant funded by the Ministry of Science and ICT (No.~RS-2025-00563533).
We thank the Institute for Basic Science (IBS-R029-C1) for funding and hosting the research visit where this work was conducted.}

\nolinenumbers 

\EventEditors{John Q. Open and Joan R. Access}
\EventNoEds{2}
\EventLongTitle{42nd Conference on Very Important Topics (CVIT 2016)}
\EventShortTitle{CVIT 2016}
\EventAcronym{CVIT}
\EventYear{2016}
\EventDate{December 24--27, 2016}
\EventLocation{Little Whinging, United Kingdom}
\EventLogo{}
\SeriesVolume{42}
\ArticleNo{23}

\begin{document}

\maketitle

\begin{abstract}
Laminar set systems consist of non-crossing subsets of a universe with set inclusion essentially corresponding to the descendant relationship of a tree, the so-called laminar tree. Laminar set systems lie at the core of many graph decompositions such as modular decompositions, split decompositions, and bi-join decompositions.
We show that from a laminar set system we can obtain the corresponding laminar tree by means of a monadic second order logic (\MSO) transduction.
This resolves an open question originally asked by Courcelle and is a satisfying resolution as \MSO is the natural logic for set systems and is sufficient to define the property ``laminar''.
Using results from Campbell et al. [STACS 2025], we can now obtain transductions for obtaining modular decompositions, co-trees, split decompositions and bi-join decompositions using \MSO instead of \CMSO.
We further gain some insight into the expressive power of counting quantifiers and provide some results towards determining when counting quantifiers can be simulated in \MSO in laminar set systems and when they cannot.

\end{abstract}
\newpage

\section{Introduction}
\label{sec:intro}

A transduction transforms relational structures over an input vocabulary $\Sigma_1$ into 
relation structures over an output vocabulary $\Sigma_2$ in such a way that each $\Sigma_2$-definable property can be ``translated'' back to a $\Sigma_1$-definable preimage.
For example, given a rooted tree, we can define a predicate $\Phi_{\SET}(S)$
    that is true for a set $S$ when it consists of all the leaves that are the descendents of a node of the tree. We can use this to construct a transduction that takes a tree and outputs a set system on the leaves of a tree; the predicate $\Phi_{\SET}(S)$ is true when $S$ is in this ``laminar'' set system, which we define below.

\begin{theorem}[{Backwards Translation Theorem~\cite[Theorem~1.40]{CE09}}]\label{thm:backwards_translation}
    Let $\tau$ be a \MSO-transduction with input vocabulary $\Sigma_1$ and output vocabulary $\Sigma_2$.
    If $\phi_2$ is an \MSO-sentence over $\Sigma_2$, then there is a \MSO-sentence $\phi_1$ over $\Sigma_1$ so that, 
    the sentence $\phi_1$ holds for precisely the structures $\mathbb{A}$ over $\Sigma_1$ for which $\tau(\mathbb{A})$ contains a structure satisfying $\phi_2$.
\end{theorem}
This is incredibly useful as it means derived structure can be used to define properties.
Of particular significance is derived tree-structure,
due to connections to tree-automata.
Thatcher and Wright showed that a property of labelled rooted trees of bounded degree is \MSO-definable if and only if it is recognizable by a tree-automaton.
So if we have a class of objects from which we can transduce a bounded tree description, from whence we can, in turn, transduce the original object,
then \MSO-definability of the original object corresponds to recognizability with a tree-automaton. 
Most prominently, for graphs of bounded treewidth, \MSO-definability is equivalent to tree-automaton recognizability,  where the forwards implication is Courcelle's Theorem \cite{TMSOLOG1} and the backward direction was finally shown by Bojańczyk and Pilipczuk in 2016 \cite{BojanczykP16}. Using the powerful method of Simon's factorization \cite{DBLP:journals/tcs/Simon90}, results of Bojańczyk et al. \cite{BojanczykGP21} show that definability equal recognizability in graphs of bounded linear clique-width.
This was recently generalized to show definability equal recognizability in classes of 
represented matroids~\cite{Boj23} and more generally representable matroids~\cite{CampbellGKKO25} over finite fields of bounded linear branch-width.
In pursuit of understanding when definability equals recognizability, Courcelle studied how to transduce tree-structured decompositions such as modular decompositions and split decompositions \cite{TMSOLOF5,Courcelle99,Courcelle06,Courcelle13}. In this series of work, Courcelle uses order invariant \MSO, which allows the use of a global linear order of the vertices of the graph. This logic is strictly more expressive than \MSO with counting quantifiers, denoted \CMSO, which in turn is more expressive than \MSO \cite{GanzowR08}.
These transductions were recently improved to avoid the use of linear order, but still relied on counting quantifiers for the key step~\cite{CampbellGKKK25}.
This key step involves transducing a tree from a collection of sets and was independently described by Bojańczyk~\cite{Boj23}, also using counting quantifiers.
It was asked as an open question in \cite{Boj23,CampbellGKKK25,CourcelleT12} whether this transduction can be obtained using just \MSO.
\\

We focus on transducing a tree from the setting of set systems as it provides a general form of structure itself, to which many other settings have a transduction to.
A \emph{set system} is a pair $(U,\mathcal{F})$ consisting of a set $U$ and a family $\mathcal{F}$ of subsets of $U$. We model this over the vocabulary $\{\SET\}$, consisting of the unary set-predicate $\SET$, by taking $\{\SET\}$-structure $\mathbb{F}$ with universe $U_{\mathbb{F}}$ and taking $\SET_{\mathbb{F}}(X)$ to be true precisely when $X\in\mathcal{F}$.
We say a set system is \emph{laminar} when: $U\in\mathcal{F}$, for each $u\in U$ the singleton $\{u\}\in\mathcal{F}$, and for any $F_1,F_2\in\mathcal{F}$ the intersection $F_1\cap F_2$ is one of $\emptyset,F_1,F_2$. 
For the last condition, we say that $F_1$ and $F_2$ do not cross. 

Any laminar set system $(U,\mathcal{F})$ naturally corresponds to a tree $T$, called the \emph{laminar tree} of $(U,\mathcal{F})$,  in which each node of $T$ corresponds to a set in $\mathcal{F}$ and the descendant relation coincides with the subset relation.
In particular, $U\in\mathcal{F}$ corresponds to the root of $T$ while
the leaves of $T$ correspond to the singletons $\{u\}\in \mathcal{F}$. For ease of notation, we identify $u\in U$ with the leaf that corresponds to the singleton $\{u\}$.
Doing this, a set $F\in\mathcal{F}$ corresponds to the node $v_F$ of $T$ where the leaves that are descendants of $v_F$ are precisely the elements of $F$.
For an example of a laminar set system and its laminar tree see \cref{fig:laminarTree}.
\\

In this paper, we give an \MSO-transduction that, given any laminar set system $(U,\mathcal{F})$, outputs its laminar tree. 
Observe that the tree structure of the laminar tree is implicitly already given by the subset relationship.
But in the set system language we do not have any means of talking about the nodes of the laminar tree, and can only specify sets of leaves.
Indeed, the tricky part is to identify elements that can play the role of each node of the laminar tree $T$. We introduce a new technique for choosing a representative leaf for every node of a tree in such a way that each leaf is used only a bounded number of times which is the key step in our transduction.
Formally, our main theorem is the following:
\begin{theorem}\label{thm:main}
There is an \MSO-transduction $\tau$ from structures over the vocabulary $\{\SET\}$ of set systems, to structures over the vocabulary $\{\descendantT\}$ of rooted trees, such that if $(U,\mathcal{F})$ is a laminar set system, then every $\tau$-image of $(U,\mathcal{F})$ is isomorphic to the laminar tree of $(U,\mathcal{F})$ as a $\descendantT$-structure.
\end{theorem}
\begin{figure}
    \includegraphics{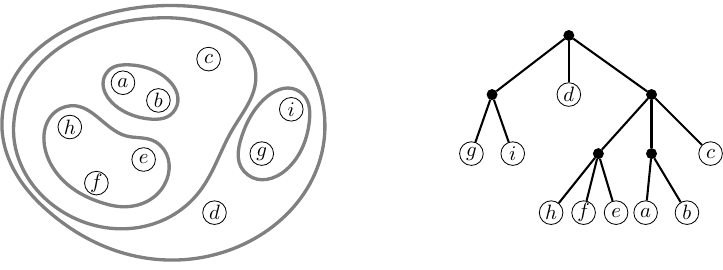}%
    \caption{%
        An example of a laminar set system and its laminar tree. %
    }%
    \label{fig:laminarTree}
\end{figure}

We present two consequences of this theorem.
Firstly,
we gain some insights into when \CMSO for laminar set systems is strictly more expressive than \MSO and when they are equivalent.

\begin{theorem}\label{bounded_d_counting}
    Let $k$ be a positive integer.
    Let $\mathcal{L}_k$ be the collection of laminar set systems whose laminar trees have down-degree at most $k$.
    Then there is a unary set-predicate $\op{EVEN}_k$, that is $\{\SET\}$-definable over \MSO, where for any set system $(U,\mathcal{F})\in\mathcal{L}_k$ and any $X\subseteq U$, we have $\op{EVEN}_k(X)$ as true if $|X|$ is even.
\end{theorem}
By adapting a known result of non-definability in \MSO,
we give the following partial converse to Theorem~\ref{bounded_d_counting}.
\begin{theorem}\label{unbounded_stars}
    If $\cL$ contains laminar set systems whose laminar trees are stars of unbounded degree, then we cannot $\{\SET\}$-define $\op{EVEN}$ for $\cL$ over \MSO.
\end{theorem}

Secondly, \cref{thm:main} has numerous corollaries in various graph theory decompositions. This work is a strengthening of Theorem~2 from~\cite{CampbellGKKK25}, which provides a transduction from laminar set systems to laminar trees using \MSO with modulo counting quantifiers. Our result shows that we can achieve the same transduction without the use of a parity predicate.
So by combining Theorem~\ref{thm:main} with the other transductions given in~\cite{CampbellGKKK25} (see \cite[Theorem~1]{CampbellGKKK25}), we obtain similar corollaries, but now without the use of counting predicates.

Weakly-partitive set systems are less restrictive than laminar set systems. In a weakly-partitive set system, two sets $F_1$ and $F_2$ in the family can cross but only if $F_1\cap F_2,F_1\setminus F_2,F_2\setminus F_1$ and $F_1\cup F_2$ are also in the family. For example, in graph theory the set of modules of a directed graph is weakly-partitive. The ``weakly-partitive tree'' of a weakly-partitive set system is obtained from the laminar tree $T$ of a particular laminar sub-set-system (the sets that do not cross any other),
by adding a node labelling $\lambda$ and a partial ordering $<$ in a way that fully capture the structure of the weakly-partitive set system (see \cite{chein1981partitive}). More specifically, $\lambda: V(T)\rightarrow \{\PRIME,\DEGENERATE,\LINEAR\}$ such that for every inner node $t\in V(T)$ with children $s_0,\dots, s_\ell$ the following holds:
\begin{itemize}
    \item if $\lambda(t)=\PRIME$ then the set of leaves below $t$ forms a set in the set system but for no set $I\subseteq [\ell]$ with $2\leq |I|< \ell$ the set of leaves below either of the $s_i$, $i\in I$ is in the set system.
    \item if $\lambda(t)=\DEGENERATE$ then the set of leaves below $t$ forms a set in the set system and for every set $I\subseteq [\ell]$ the set of leaves below either of the $s_i$, $i\in I$ is in the set system.
    \item if $\lambda(t)=\LINEAR$ then the set of leaves below $t$ forms a set in the set system. It further must hold that $s_0<\dots< s_\ell$ and  for a set $I\subseteq [\ell]$  the set of leaves below either of the $s_i$, $i\in I$ is in the set system exactly when $I$ forms a $<$-interval.
\end{itemize}
Furthermore, every set in the bipartitive set system is of this form.
We obtain the following transduction.
\begin{corollary}\label{cor:partitive}
    There is an \MSO-transductions $\tau$ such that given  any weakly-partitive set system $(U,\mathcal{S})$ modelled by the $\{\SET\}$-structure $\mathbb{S}$, $\tau$ outputs the weakly-partitive tree $(T,\lambda,<)$ of $(U,\mathcal{S})$.
\end{corollary}
Instead of considering set systems, we can also consider systems of bi-partitions. Systems of bi-partitions naturally model concepts in graph theory such as splits and bi-joins.  The concept analogous to weak-partitiveness in the world of systems of bipartitions is weak-bipartitiveness. We obtain a similar corollary which allows one to \MSO-transduce a tree and additional structure that captures the bipartitions of a weakly-bipartitive system.  
\begin{corollary}\label{cor:biPartitive}
    There is an \MSO-transductions $\tau$ such that, given any weakly-bipartitive system of bipartitions  $(U,\mathcal{B})$ modelled by the $\{\BIPART\}$-structure $\mathbb{B}$, $\tau$ outputs the weakly-bipartitive tree $(T,\lambda,<)$ of $(U,\mathcal{B})$.
\end{corollary}
Finally, we can apply \cref{thm:main} and the machinery from \cite{CampbellGKKK25} to obtain transductions for the graph decompositions mentioned above whose underlying set system or system of bipartitions are weakly-partitive or weakly-bipartitive, repectively.
\begin{corollary}\label{cor:graphDecompositions}
    There are non-deterministic \MSO-transductions $\tau_1,\ldots,\tau_4$ such that: 
    \begin{enumerate}
        \item 
            Given any (directed) graph $G$ as input, $\tau_1$ outputs the modular decomposition $(T,F)$ of $G$.
        \item 
            Given any (directed) cograph $G$, $\tau_2$ outputs the cotree $(T,\lambda)$ of $G$.
        \item 
             Given any (directed) graph $G$, $\tau_3$ outputs a split decomposition $(T,F)$ of $G$.
        \item 
             Given any graph $G$, $\tau_4$ outputs a bi-join decomposition $(T,F)$ of $G$. 
    \end{enumerate}
\end{corollary}
In \cref{sec:beyondLaminar}, we discuss in further detail how to modify the proofs from~\cite{CampbellGKKK25} to obtain \cref{cor:partitive,cor:biPartitive,cor:graphDecompositions}.

We believe that this line of research of considering transductions for simple, but dense tree-structured decompositions is important in the context of definability-equals-recognizability results. Notably, there are such results for trees \cite{ThatcherWright} and graphs of bounded treewidth \cite{BojanczykP16} which we can think of as sparse tree-structured decompositions. Additionally, there are definability-equals-recognizability results for dense graph classes in the case the decomposition used has a linear structure \cite{BojanczykGP21,CampbellGKKO25}. It is open, however, whether in graphs of bounded clique-width, for which \cite{DBLP:journals/mst/CourcelleMR00} yields one direction, whether definability equals recognizability. In some sense, our results solve a restricted version of this major open problem.

\section{Preliminaries}
\label{sec:prelim}
We denote by $[n]$ for $n\in \mathbb{N}$ the set $[n]=\{1,\dots, n\}$. A \emph{partition} $(P^1,\dots, P^\ell)$ of a set $U$ has to satisfy that the parts $P_i$ are pairwise disjoint and their union is $U$ but we do not require parts to be non-empty. For a partition $(P^1,\dots, P^\ell)$ a \emph{refinement} of $(P^1,\dots, P^\ell)$ is a partition $(S^1,\dots, S^k)$ such that each $S^i\subseteq P^j$ for some $j\in [\ell]$.

We consider trees to be rooted and, therefore, equipped with an ancestor-descendant relationship. We use the usual terminology regarding \emph{parent, child} and \emph{sibling}. We distinguish between \emph{ancestors/descendants} of a node $t$ that include $t$ itself, and \emph{proper ancestors/proper descendants} of $t$ that are all ancestors/descendants except $t$.
In addition, for a tree $T$ we denote the set of nodes by $V(T)$ and the set of edges of $T$ by $E(T)$.
For a tree $T$ nodes of degree $1$ are called \emph{leaves} of $T$.
Each node in $T$ that is not a leaf, we call an \emph{inner node}.
For a tree $T$ and a node $t\in V(T)$, we denote by $T_t$ the subtree of $T$ induced by all descendants of $t$. 

\subsection{Logic and transductions}
In order to model set systems we consider extended relational structures which extend the usual notion of relational structures to allow the use of set predicates. An \emph{extended relational vocabulary} is a set $\Sigma$ of symbols  each having an associated arity, denoted $\arity(Q)$ for $Q\in \Sigma$, where each symbol is either a relation name or a set predicate name.  An \emph{extended relational structure} over $\Sigma$, or short \emph{$\Sigma$-structure}, is a tuple $\mathbb{A}=(U_\mathbb{A},(Q_\mathbb{A})_{Q\in \Sigma})$ consisting of a finite set of elements $U_\mathcal{A}$, called the \emph{universe} of $\mathbb{A}$, and an interpretation $Q_{\mathbb{A}}$ of every $Q\in \Sigma$ where $Q_{\mathbb{A}}$ is a relation over $U_{\mathbb{A}}$ if $Q$ is a relational name and a relation over subsets of $Q_{\mathbb{A}}$ if $Q$ is a set predicate name. For two vocabularies $\Sigma$ and $\Gamma$, and a $\Sigma$-structure $\mathbb{A}$ and a $\Gamma$-structure $\mathbb{B}$ we let  the \emph{union}, denoted $\mathbb{A}\sqcup \mathbb{B}$, be the $\Sigma\cup \Gamma$-structure with universe $U_\mathbb{A}\cup U_\mathbb{B}$ and  the interpretation $Q_{\mathbb{A}\sqcup\mathbb{B}}$ of some symbol~$Q\in\Sigma\cup\Gamma$ being either~$Q_{\mathbb{A}}$, $Q_{\mathbb{B}}$ or~$Q_{\mathbb{A}}\cup Q_{\mathbb{B}}$
depending on whether~$Q$ belongs to~$\Sigma\setminus\Gamma$, to~$\Gamma\setminus\Sigma$, or to~$\Sigma\cap\Gamma$.
For a $\Sigma$-structure~$\mathbb{A}$
and a $\Gamma$-structure~$\mathbb{B}$,
we write~$\mathbb{A}\sqsubseteq\mathbb{B}$
if~$\Sigma\subseteq\Gamma$,
$U_{\mathbb{A}}\subseteq U_{\mathbb{B}}$
and for each symbol~$Q$ in~$\Sigma$,
$Q_{\mathbb{A}}=Q_{\mathbb{B}}$.

We define syntax and semantic of \emph{monadic second order logic}, short \MSO, in the usual way and refer, for example, to
\cite{CE09,FunkMN22,Hlineny06,Strozecki11} for the definition of \MSO on extended relational structures. \MSO extends first-order logic introducing set variables and allowing quantification over set variables. To clearly identify the type of variables, we distinguish between \emph{element variables} and \emph{set variables} and use 
 lowercase letters, such as $r,x,y,z,\dots$, for element variables  and uppercase letters, such as $R,X,Y,Z,\dots$, for set variables. 
For a formula~$\phi$, we write $\phi(x_1,\dots,x_\ell,X_1,\dots, X_k)$ to indicate that the free element variables of $\phi$ are $x_1,\dots, x_\ell$ and the free set variables of $\phi$ are $X_1,\dots, X_k$,
namely, the set of variables occurring in~$\phi$ that are not bound to a quantifier within~$\phi$.
A \emph{sentence} is a formula without any free variables. Where possible, we  follow the convention that relational names and formulas with only free element variables are lowercase while set predicate names as well as formulas that have only free set variables receive uppercase names.

We use the following extended relational structures to model the structures used in this paper. We use vocabulary $\{\SET\}$ where set is a set predicate name of arity $1$ to model set systems. Naturally a set system $(U,\mathcal{F})$ is modelled by the $\{\SET\}$-structure $\mathbb{F}$ with universe $U_\mathcal{F}=U$ and for which $\mathcal{F}$ is the interpretation $\SET_\mathbb{F}$. Additionally, we use vocabulary $\{\descendant\}$ to model trees. Here a tree $T$ is modelled by the $\{\descendant\}$-structure $\mathbb{T}$ with universe $U_\mathbb{T}=V(T)$ and the interpretation $\descendant_\mathbb{T}$ being the descendant relationship of $T$. Note that choosing the descendant relationship is arbitrary and the relations $\ancestor, \child$ and $\parent$ can easily be defined in $\MSO$.

We note that \MSO is sufficient to define laminar sets systems. 
Indeed, laminar set systems are exactly those that satisfy the $\{\SET\}$-formula
\begin{align*}
&\Big(\exists U\forall x (x\in U\vee \SET(U))\Big)\land \Big(\forall x\exists S \forall y (\SET(S)\land\left[y\in S\leftrightarrow y=x\right])\Big)\land\\
&\qquad\Big(\forall F_1\forall F_2(\SET(F_1)\wedge\SET(F_2)\rightarrow 
\left[F_1\cap F_2=\emptyset\vee F_1\cap F_2=F_1\vee F_1\cap F_2=F_2\right])\Big).
\end{align*}

\paragraph*{Transductions}
Let~$\Sigma$ and~$\Gamma$ be two extended relational vocabularies.
A \emph{$\Sigma$-to-$\Gamma$ transduction}
is a set~$\tau$ of pairs formed by
a $\Sigma$-structure, call the \emph{input},
and a $\Gamma$-structure, called the \emph{output}.
We write \emph{$\mathbb{B}\in\tau(\mathbb{A})$} when~$(\mathbb{A},\mathbb{B})\in\tau$.
When for every pair~$(\mathbb{A},\mathbb{B})\in\tau$
we have~$\mathbb{A}\sqsubseteq\mathbb{B}$,
we call~$\tau$ an
\emph{overlay transduction}.
Transductions
are defined over some logic $\mathcal{L}$, such as \FO, \MSO or \CMSO. In  this paper the ambient logic is \MSO and we therefore restrict the following definitions to this case. 
We define an \emph{\MSO-transduction}
to be a transduction obtained
by composing a finite number of \emph{atomic \MSO-transductions}
of the following kinds.
\begin{description}
    \item[Colouring]
        A \emph{$k$-colouring tansduction}  adds $k$ new unary predicates~$\colorV_i$ for every $i\in [k]$ to the signature $\Sigma$ while adding $k$ new unary relations to the original $\Sigma$-structure $\mathcal{A}$ interpreting the $k$ additional predicates.
        Any possible interpretation yields an output of the transduction.
        Hence, it defines a total (non-functional) relation
        from $\Sigma$-structures to~$\Gamma$-structures
        where~$\Gamma=\Sigma\cup\{\colorV_i\mid i\in [k]\}$.
    \item[Copying]
        A \emph{$k$-copying transduction} adds $k$ copies of the universe of the original  $\Sigma$-structure $\mathbb{A}$ to  $\mathbb{A}$ while introducing new binary predicates $\copyV_i(x,y)$ for every $i\in [k]$ which expresses that element $x$ is the $i$-th copy of the original element $y$. Hence, this defines a function from~$\Sigma$-structures to~$\Gamma$-structures,
        where~$\Gamma=\Sigma\cup\{\copyV_i\mid i\in[k]\}$. For sake of easier notation, we assume that a binary predicate $\copyV_0(x,y)$ is added which is satisfied for elements $x,y$ if $x$ is an original element and $y=x$ ($y$ is the original copy of $x$).
    \item[Filtering]
        A \emph{filtering transduction} is specified by a sentence $\chi$ over $\Sigma$ and outputs the original $\Sigma$-structure $\mathbb{A}$ if $\mathbb{A}$ satisfies $\chi$.
        Hence, a filtering transduction defines a partial function from $\Sigma$-structures to $\Sigma$-structures.
    \item[Interpretation] 
        An \emph{interpretation} is a transduction from $\Sigma$ to $\Gamma$-structure for any signature $\Sigma$ and $\Gamma$. An interpretation is specified by a tuple $(\phi,(\psi_Q)_{Q\in \Gamma})$ where $\phi$ is a formula over $\Sigma$ with one free variable and $\psi_Q$ for any $Q\in \Gamma$ is a formula over $\Sigma$ with $\arity(Q)$ free variables (which may be set variables if $Q$ is a set predicate). For input $\Sigma$-structure $\mathbb{A}$ the interpretation returns the $\Gamma$-structure whose universe consists of all elements of the universe of $\mathbb{A}$ satisfying $\phi$ and for which $Q$ is interpreted by the set of all tuples satisfying $\psi_Q$ for each $Q\in \Gamma$.
        This defines a function from $\Sigma$-structures to $\Gamma$-structures.
\end{description}
We often refer to unary predicates as colours and therefore we describe $k$-colouring transductions as a colouring of the universe with $k$ colours. Note that in such a colouring any element might receive multiple colours or even no colour at all.

In the following we show that we can easily obtain laminar set systems from their laminar trees by means of a transduction.
\begin{lemma}\label{lem:backwardsDirection}
    For any laminar set system $(U,\mathcal{F})$ with laminar tree $T$ modelled by the $\{\descendant\}$-structure $\mathbb{T}$ there is an \MSO-transduction $\tau$ which on input $\mathbb{T}$ produces the $\{\SET\}$-structure $\mathbb{F}$ that models $(U,\mathcal{F})$.
\end{lemma}
\begin{proof}
    We can transduce the set system from the laminar tree using a single \MSO-interpretation.
    We obtain the $\{\SET\}$-structure $\mathbb{F}$ modelling $(U,\mathcal{F})$ from $\mathbb{T}$ by restricting the universe consisting of all nodes of $T$ to just the leaves of $T$. Additionally, we can define the predicate $\SET$ by a formula $\psi_{\SET}(X)$ expressing that $X$ is the set of leaves of  the subtree $T_t$ of $T$ rooted at $t$ for some node $t\in V(T)$. Specifically, we can implement this interpretation as follows:
    \begin{align*}
        \phi(x)&:=\forall y\big(\descendant(y,x)\rightarrow y=x\big), \text{ and }\\
        \psi_{\SET}(X)&:= \exists y\forall x\big(\descendant(x,y)\leftrightarrow x\in X\big).
    \end{align*}
\end{proof}
Note by the Backwards Translation Theorem (\cref{thm:backwards_translation}), this means that any property of laminar set system that can be defined in the language of sets corresponds to a property of rooted trees definable in the language of rooted trees.

\section{Transducing laminar trees}\label{sec:mainTransduction}
\label{sec:transduction} 
In this section we prove our main result, which consists in constructing a transduction that, given a laminar set system $(U,\mathcal{F})$ outputs its laminar tree $T$.

Our first task is to choose an elements of the set system for every inner node. This element will play the role of the node in the tree structure. However, we need to guarantee that every element plays the role of at most a bounded number of inner nodes (in our case this bound is $17$) and we can identify which element plays the role of a particular node in \MSO.
In order to achieve this, we partition the nodes of $T$ into $17$ sets.
For each of these parts, we will use a copy of the universe $U$ to contain elements to play the roles of the nodes in this part.
We will call the element (in a copy of $U$) that plays the role of some node $t$ of $T$ the ``realization'' of the set $F\in \mathbb{F}$ corresponding to the node $t$.
Any copy of an element can be the realization of at most one set.
One of the $17$ parts will contain exactly the leaves of the tree $T$ for which it is trivial to identify realizations.  All other parts of the partition consist of a carefully chosen set of inner nodes. From now on, assume that $S$ is one of the parts containing only inner nodes of $T$. To find realizations for all nodes in $S$, we first choose a subtree of $T$ for each element $s\in S$, denoted $H_s$. We show that we can choose such trees to be pairwise disjoint if we choose $S$ to be thin.
Thinness roughly enforces that we can always grow the representative tree of a node $t\in S$ downwards in such a way that avoids $S$ itself, which is crucial for guaranteeing disjointness of representative trees. 

Besides choosing the representative trees disjoint, we also enforce them to have a particular structure. This structure allows us, given all the leaves of representative trees in $S$ to pinpoint (by means of an \MSO-formula) exactly which leaves belong to the same representative tree. Hence, we can use a colouring-transduction to identify all possible choices of the leaves of the representative trees of nodes in $S$. We then filter which colourings correctly correspond to the set of leaves of representative trees. We chose an arbitrary leaf in the set of leaves of the representative tree of a node $t\in S$ to be the realization of the set $F$ corresponding to $t$. Distinctness trivially follows from the disjointness of the representative trees.  We then interpret, restricting the universe to only contain realizations of sets and introducing the descendant relationship through the subset relationship.

\subsection{Representative sets}
Let $(U,\mathcal{F})$ be a laminar set system, $T$ be the laminar tree of $(U,\mathcal{F})$ and $L$ the set of leaves of $T$. We equip $T$ with a labelling $\lvlParity_T:V(T)\rightarrow \{0,1,2,3\}$ where a node of depth $d$ in $T$ is assigned label $d\mod 4$. 

Let $s$ be a node of $T$. 
A \emph{representative tree} of $s$ is a subtree $H$ of $T$ with the following properties. 
\begin{description}   
    \item[(RT1)] $s$ is the root of $H$.
    \item[(RT2)] If $\lvlParity_T(s)\equiv \lvlParity_T(t)\mod 2$ for some $t\in V(H)$, then either $t$ is a leaf of $T$ or $H$ contains exactly one child of $t$.
    \item[(RT3)]  If $\lvlParity_T(s)\not\equiv \lvlParity_T(t) \mod 2$ for some $t\in V(H)$, then either $t$ is a leaf of $T$ or $H$ contains every child of $t$.
\end{description}
For a set $B\subseteq L$ we call the parent of the least common ancestor of the leaves in $B$ the \emph{tip} of $B$. We call the minimum subtree of $T$ which contains both $B$ and the tip of $B$ the \emph{up-tree} of $B$, denoted $H_{B}^\uparrow$.
A \emph{representative set} of a node $s$ of $T$ is a set $A_s$ of leaves of $T$ for which $s$ is the tip and the up-tree $H_{A_s}^\uparrow$ of $A_s$ is a representative tree. 
  
A set $S\subseteq V(T)\setminus L$ of inner nodes of $T$ is called \emph{thin} if there is $i\in \{0,1,2,3\}$ such that the following properties hold.
\begin{description}
    \item[(TS1)] $\lvlParity_T(s)=i$ for every $s\in S$.
    \item[(TS2)] Every $s\in S$ has a sibling which is not in $S$ with the exception of the root of $T$.
    \item[(TS3)] The parent $t_s$ of every $s\in S$ has a sibling $t'_s$ such that no child of $t'_s$ is in $S$ with the exception of nodes $s\in S$ that are children of the root of $T$.
\end{description}
We first prove that thinness is enough to choose disjoint representative trees for all nodes of the thin set. For this, the following observation is an important ingredient.
\begin{observation}\label{obs:atLeastTwoChildren}
    In the laminar tree $T$ of any laminar set system $(U,\mathcal{F})$ every inner node has at least two children.
\end{observation}
Note that for the following lemma property (TS2) is not needed. We, however, require this property later. 
\begin{figure}
    \includegraphics[width=\textwidth]{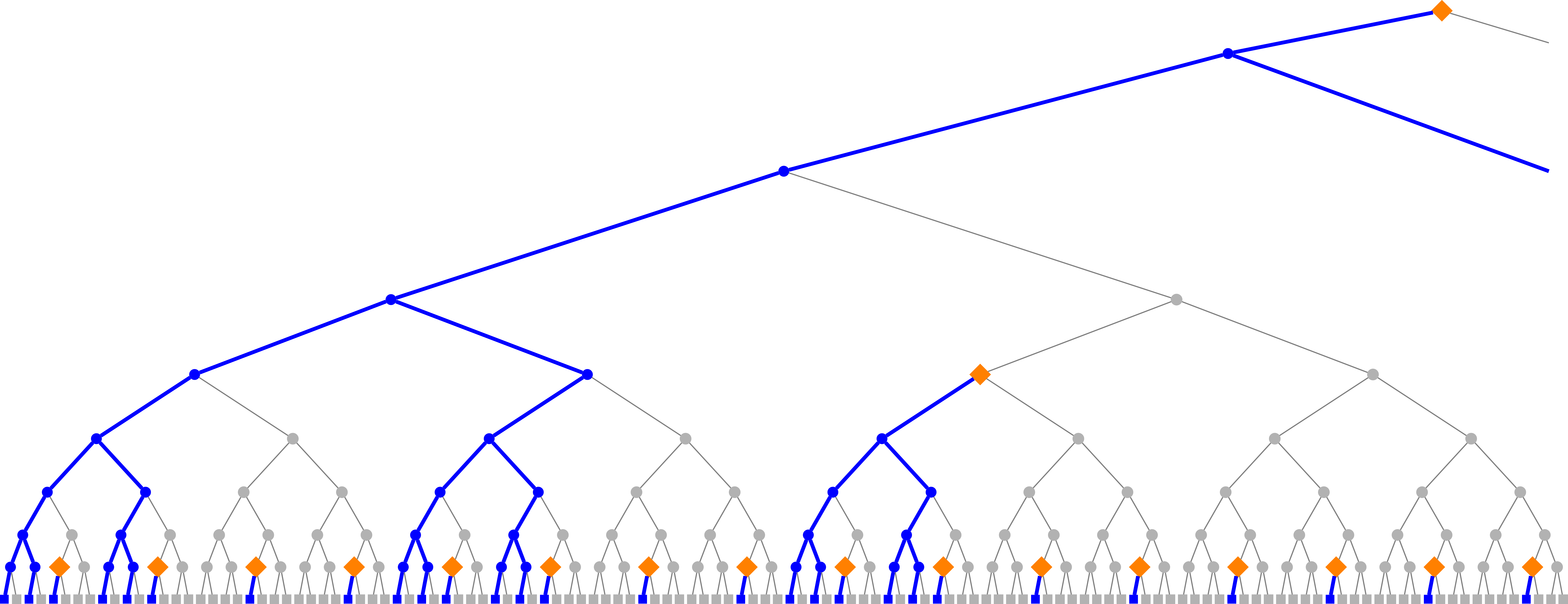}%
    \caption{%
        Part of a laminar tree with thin set $S$ consisting of the larger orange nodes and pairwise disjoint representative trees of the nodes in $S$ highlighted in blue. %
    }%
    \label{fig:representativeSets}
\end{figure}

\begin{lemma}\label{lem:disjointRepSets}
    If $S\subseteq V(T)\setminus L$ is thin, then we can choose representative sets $(A_s)_{s\in S}$ such that the representative trees $H_{A_s}^\uparrow$, $s\in S$ are pairwise disjoint.
\end{lemma}
\begin{proof} 
    Let $S\subseteq V(T)\setminus L$ be a thin set and let $i\in \{0,1,2,3\}$ be the index such that $\lvlParity_T(s)=i$ for every $s\in S$ which exists by (TS1). Let $P$ be the set of all nodes which are parents of a node in $S$. Note that every $p\in P$ satisfies $\lvlParity_T(p)\equiv i-1\mod 4$. We show that for every $s\in S$ we can choose a representative tree $H_s$, which avoids all nodes in $P$ (i.e. does not contain any node in $P$). This statement implies that the representative trees $H_s$ are pairwise disjoint. Clearly, since $H_s$ is a subtree of $T_s$ for every $s\in S$ by (RT1), $H_s$ and $H_{s'}$ must be disjoint whenever $s,s'\in S$ are not in any ancestor-descendant relationship. If, on the other hand, $s\in S$ is a proper ancestor of $s'\in S$ then $H_s$ cannot intersect $T_{s'}$ (and hence cannot intersect $H_{s'}$) since $H_s$ avoids the parent of $s'$. 

    We now describe how to construct $H_s$ which avoids $P$ for some fixed $s\in S$. We construct the tree $H_s$ recursively, level by level. The $0$-th level of $H_s$ consists of just $s$, which clearly avoids $P$. Assume that we have constructed $j$ levels of $H_s$ that avoid $P$ and aim to construct the $(j+1)$-th level. We distinguish the following cases. First, assume $j\equiv i \mod 4$. In this case, for every node $t$ on level $j$ of $H_s$, which is not a leaf, we choose an arbitrary child $t'$ of $t$ and add it to $H_s$. Note that the updated tree $H_s$ clearly still avoids $P$ as $\lvlParity_T(t')\equiv i+1 \mod 4$ for every newly added node $t'$ while nodes $p\in P$ satisfy $\lvlParity_T(p)\equiv i-1 \mod 4$. Next, assume that $j\equiv i+1 \mod 4$. In this case, for every $t$ on level $j$ of $H_s$, which is not a leaf of $T$, we add all children of $t$ to $H_s$. Again, each newly added vertex $t'$ cannot be contained in $P$ as they satisfy $\lvlParity(t')\equiv i+2 \mod 4$. Next, we assume that $j\equiv i+2 \mod 4$. In this case, for every node $t$ on level $j$ of $H_s$, which is not a leaf, we choose a child $t'$ of $t$ which is not in $P$ and add it to $H_s$. This choice is always possible, as $S$ is thin and, therefore, every node $p\in P$ must have a sibling that does not have children in $S$ by (TS3) which implies that this child is not contained in $P$. Clearly, the updated tree $H_s$ avoids $P$. Finally, assume $j\equiv i+3 \mod 4$. For every node $t$ on level $j$ of $H_s$, which is not a leaf, we add all its children to $H_s$. The resulting tree $H_s$ still must avoid $P$ as each newly added node $t'$ satisfies $\lvlParity_T(t')=i$ while nodes  $p\in P$ satisfy $\lvlParity_T(p)\equiv i-1 \mod 4$. Note that the tree $H_s$ is a representative tree of $s$ by construction because $s$ is the root of $H_s$ implying (RT1),  for every node $t$  on level $j$  with $j\equiv i \mod 2$ (unless $t$ is a leaf) we added exactly one child of $t$ to $H_s$ implying (RT2) and for every node $t$  on level $j$ with $j\not\equiv i \mod 2$ (unless $t$ is a leaf) we added all children of $t$ to $H_s$ implying (RT3). For an illustration of the construction of representative trees of a thin set $S$ we refer the reader to \cref{fig:representativeSets}.

    We can now set $A_s$ to be the leaves of $H_s$ for every $s\in S$. We argue that $H^\uparrow_{A_s}=H_s$. First note that because $s$ is the root of $H_s$ the least common ancestor of the leaves in $A_s$ is some descendant of $s$. By construction, there is exactly one child $t$ of $s$ which is contained in $H_s$ and therefore the least common ancestor of the leaves in $A_s$ is indeed a descendant of $t$. Let $t_1,\dots, t_\ell$ be the children of $t$ and note that $\ell\geq 2$ by \cref{obs:atLeastTwoChildren}. By construction $t_1,\dots, t_\ell$ are contained in $H_s$ (as $\lvlParity(t)\equiv i+1\mod 4$). Since both the subtree of $H_s$ rooted at $t_1$ and the subtree of $H_s$ rooted at $t_2$ must have a leaf and therefore an element of $A_s$, the least common ancestor of the leaves in $A_s$ is $t$. Therefore the tip of $A_s$ is $s$. Hence, by definition of up-tree, we know that $H^\uparrow_{A_s}$ is a subtree of $H_s$ as $H_s$ contains all nodes in $A_s$ and the tip of $A_s$. Furthermore, every node $t$ contained in $H_s$ must be in $H^\uparrow_{A_s}$ because the subtree of $H_s$ rooted at $t$ must contain a leaf $t'$ which, by construction, is in $A_s$. Then the path from $t'$ to $s$ must be contained in $H^\uparrow_{A_s}$ implying that $t$ is in $H^\uparrow_{A_s}$.  Since $H^\uparrow_{A_s}=H_s$ is a representative tree, $A_s$ is a representative set. Hence, we have chosen representative sets $(A_s)_{s\in S}$ with pairwise disjoint representative trees $H^\uparrow_{A_s}$ proving the statement.
\end{proof}

We now aim to prove that given the union $A:= \bigcup_{s\in S} A_s$ of representative sets of a thin set, the representative sets are uniquely defined. We further aim to characterize representative sets in such a way that we can later define them in \MSO. For this we use the following terminology.
Given a set $B$ of leaves of $T$ we say that an inner-node $s\in V(T)\setminus L$ is \emph{fully-branched} (by $B$) if for every child $t$ of $s$ the set of leaves in the subtree of $T$ with root $t$ contains a leaf from $B$. Additionally, we say that $s$ is \emph{single-branched} (by $B$) if there is exactly one child $t$ of $s$ for which the subtree of $T$ rooted at $t$ contains a leaf from $B$. We say that $s$ is \emph{missed} (by $B$) if the subtree of $T$ rooted at $s$ contains no leaves from $B$.

\begin{lemma}\label{lem:uniqueRepSets}
    Let $S$ be a thin set, $(A_s)_{s\in S}$ representative sets with pairwise disjoint up-trees $H_{A_s}^\uparrow$ and $A:= \bigcup_{s\in S} A_s$. For any $B\subseteq L$ and any node $s$ of $T$,
    $B=A_s$ if and only if all of the following hold: 
    \begin{romanenumerate}
        \item $B\subseteq A$,
        \item $s$ is the tip of $B$,
        \item if  a descendant $t$ of $s$ is fully-branched by $B$, then each child of $t$ is either a node single-branched by $B$ or a leaf in $B$,
        \item if a descendant $t$ of $s$ is single-branched by $B$ then exactly one child of $t$ is either a node fully-branched by $B$ or a leaf in $B$, while all other children of $t$ are either missed by $B$ or leaves not in $B$, and
        \item there is no proper superset $B'$ of $B$ and node $s'$ of $T$ such that $B'$ together with $s'$ satisfy properties {\normalfont(i)-(v)}. 
    \end{romanenumerate}
\end{lemma}
\begin{proof}
    We start by proving parts of the forwards direction of the statement.
    \begin{claim}\label{claim:prop(i)-(iv)}
        For every $s\in S$ the representative set $A_s$ and node $s$ satisfy conditions (i)-(iv).
    \end{claim}
    \begin{claimproof}
        First, note that condition (i) is trivially satisfied. Additionally, (ii) is satisfied by the definition of representative sets.
        By definition $H_{A_s}^{\uparrow}$ is a representative tree. 
        Observe that for any node $t$ of $H_{A_s}^{\uparrow}$ there is some leaf $v$ of $T_t$ that is contained in $H_{A_s}^{\uparrow}$ by the definition of representative trees. In particular, this implies that every node $t$ in $H_{A_s}^{\uparrow}$ with $\lvlParity_T(t)\equiv \lvlParity_T(s)\mod 2$ is either a leaf of $T$ or single-branched by $A_s$ and every node $t$ in $H_{A_s}^{\uparrow}$ with $\lvlParity_T(t)\not\equiv \lvlParity_T(s)\mod 2$ is either a leaf of $T$ or fully-branched.
        Since every inner node $t$ not contained in $H_{A_s}^\uparrow$ is missed by $A_s$, we obtain properties (iii)-(iv). 
        
    \end{claimproof} 
    The essential property we use to prove the equivalence in the statement of the lemma is that the up-tree $H_B^\uparrow$ of any set $B$ satisfying the properties (i)-(iv) together with the tip of $B$ must be a subtree of $H_{A_s}^\uparrow$ for some $s\in S$. We prove this statement in two steps \cref{claim:forestProp} and \cref{claim:subtreeOfUpwardsForest}.
    Let $F^{\uparrow}$ be the forest obtained by taking the union of $H_{A_s}^{\uparrow}$ for all $s\in S$. 
    \begin{claim}\label{claim:forestProp}
        It holds that $F^{\uparrow}$ is an induced subgraph of $T$ with connected components $H_{A_s}^\uparrow$, $s\in S$. 
    \end{claim}
    \begin{claimproof}
        Let $i\in \{0,\dots, 3\}$ such that $i=\lvlParity_T(s)$ for every $s\in S$ (note that $i$ is well defined as $S$ is thin). First note that by assumption the subtrees $H_{A_s}^\uparrow$ are pairwise disjoint and therefore are the connected components of $F^{\uparrow}$.
        Towards a contradiction assume the claim is not true and let $u,s\in V(F^{\uparrow})$ be nodes of $T$ such that $s$ is the child of $u$ in $T$, but the edge $us$ is not in $E(F^{\uparrow})$. In this case $s$ is the root of some tree of $F^{\uparrow}$ and in particular $s\in S$. Additionally, $u$ is a node in some tree $H_{\uparrow}^{u'}$ with $u'\in S$ and $u'\not= s$. 
        In particular, $\lvlParity_T(s)=\lvlParity_T(u')=i$. Since $u$ is the parent of $s$ this implies that $\lvlParity_T(u')\not\equiv \lvlParity_T(u)\mod 2$. Since additionally $H_{\uparrow}^{u'}$ is a representative tree, all children of $u$ must 
        be contained in $H_{\uparrow}^{u'}$. This contradicts the disjointness of representative trees, since $s$ is a child of $u$ and is contained in $H_{\uparrow}^s$.
    \end{claimproof}

    \begin{claim}\label{claim:subtreeOfUpwardsForest}
        For any set $B$ with tip $s$ such that $B$ and $s$ satisfy properties (i)-(iv) the tree $H_B^{\uparrow}$ does not contain any node which is not contained in $F^{\uparrow}$.
    \end{claim}
    \begin{claimproof}
        Our argument essentially amounts to showing that there are not enough leaves from $A$ in any subtree of $T$ with root in $V(H_B^{\uparrow})\setminus V(F^{\uparrow})$ in order for $B$ to satisfy properties (i)-(iv). See \cref{fig:claim14} for an illustration. 
        To argue this formally, let $t$ be a node which is contained in $H_B^{\uparrow}$ but not in $F^{\uparrow}$ and has maximum depth with this property. 
        Let $(p_0,p_1,\dots, p_{\ell})$ with $p_0=s$ and $p_{\ell}=t$ be the path in $H_B^{\uparrow}$ from $s$ to $t$. First observe that no node of $p_0,\dots, p_{\ell}$ can be missed by $B$ as they are contained in $H_B^{\uparrow}$ which implies that at least one descendant must be contained in $B$ for each of $p_0,\dots, p_{\ell}$. Property (ii) implies that $p_0=s$ is single-branched by $B$. By property (iv) and since $p_1$ cannot be missed or be a leaf, we conclude that $p_1$ is fully-branched. By property (iii) and because $p_2$ cannot be a leaf, this implies that $p_2$ is single-branched by $B$. Using a recursive argument, we conclude that $p_{i}$ must be fully-branched by $B$ if $i$ is odd and $p_i$ is single-branched by $B$ if $i$ is even for every $i\in \{0,\dots, \ell\}$. 

        We can further argue the following property:
        \begin{align*}
            (\ast) \quad \phantom{ii}&\text{For every leaf $u\in A_{t'}\cap B$ it holds that $t'$ is a child of $t$. } 
        \end{align*} 
       
        \begin{figure}
        \centering
    \includegraphics[width=\textwidth]{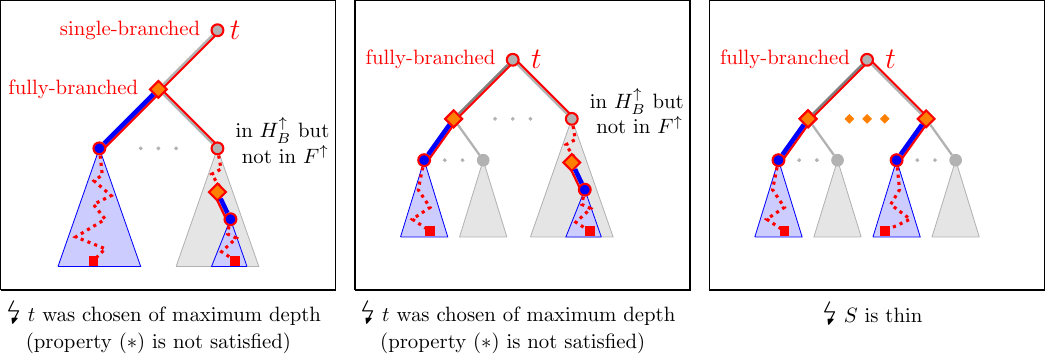}%
    \caption{%
         The different cases to obtain a contradiction in \cref{claim:subtreeOfUpwardsForest}. Here larger orange diamonds are nodes in $S$, representative trees are highlighted in blue (in particular $F^\uparrow$ is blue), parts of the tree belonging to $H_B^\uparrow$ are highlighted in red and the red squares are leaves in $B$.
    }%
    \label{fig:claim14}
\end{figure}
        To see this, observe that for any node $t'\in S$ and $u\in A_{t'}\cap B$ the path between $t'$ and $u$ must be fully contained in $F^\uparrow$. Since $t$ is not contained in $F^\uparrow$,  we conclude that $t'$ must be a proper descendant of $t$. 
        Now, let $t'\in S$ be any node that is contained in the subtree $(H_B^{\uparrow})_{t}$ of $H_B^{\uparrow}$ rooted at $t$. By \cref{claim:forestProp} we have that $H_{A_{t'}}^\uparrow$ is a connected component of $F^\uparrow$. Since $F^\uparrow$ is an induced subgraph of $T$ by \cref{claim:forestProp}, this implies that the parent $t''$ of $t'$ cannot be contained in $F^\uparrow$. On the other hand, $t''$ must be contained in $H_B^\uparrow$ as $t$ is in the tree $H_B^\uparrow$ and $t$ is an ancestor of $t''$. Since we choose $t$ of maximum depth with the property of being in $H_B^\uparrow$ but not in $F^\uparrow$ this implies that $t''=t$. Therefore, every node in $(H_B^{\uparrow})_{t}$ that is contained in $S$ must be a child of $t$ and hence property $(\ast)$ holds. 
       We now distinguish two cases. 

        First assume that $\ell$ is odd and therefore $t=p_\ell$ is fully-branched (see right two part of \cref{fig:claim14}). Hence, every child of $t$ must be the ancestor of a node in $B\subseteq A$. By property $(\ast)$, every child of $t$ must be contained in $S$. But this contradicts the fact that $S$ is thin.

        On the other hand, assume $\ell$ is even and, therefore, $t=p_\ell$ is single-branched (see left part of \cref{fig:claim14}). Let $t'$ be the child of $t$ that is either a leaf or fully branched, and which exists due to (iv). By statement $(\ast)$, we know that $t'\in S$ and therefore cannot be a leaf as $S$ only contains inner nodes of $T$. Because $t'\in S$ we know that $H^\uparrow_{A_{t'}}$ is a representative tree and therefore $t'$ has only one child contained in $H^\uparrow_{A_{t'}}$ by (RT2). By \cref{obs:atLeastTwoChildren} $t'$ has at least two children and hence there is a child $t''$ of $t'$ such that $T_{t''}$ contains no leaves from $A_{t'}$. On the other hand, by (iv) $t'$ is fully-branched by $B$ (as $t$ is single-branched) and therefore there must be a leaf $u\in B\subseteq A$ in the subtree $T_{t''}$. In particular, $u\notin A_{t'}$. But this contradicts the property $(\ast)$ as $t'$ is the child of $t$ that is an ancestor of $u$.  
    \end{claimproof}
    Using the above claims we can now conclude the proof as follows. 
    First assume $B=A_s$ for some $s\in S$. By \cref{claim:prop(i)-(iv)}, $B$ and $s$ satisfy (i)-(iv). Presume that $B$ does not satisfy (v) and $B'\supset B$ and node $s'$ satisfy properties (i)-(iv). Combining \cref{claim:subtreeOfUpwardsForest} and \cref{claim:forestProp}, we know that $H_{B'}^\uparrow$ must be a subtree of a connected component of $F^{\uparrow}$. More specifically, $H_{B'}^\uparrow$ must be a subtree of $H_{B}^\uparrow=H_{A_s}^\uparrow$. This implies that $B'\subseteq B$ which contradicts the assumption that $B'$ is a proper superset of $B$.
    
    On the other hand, assume $B\subseteq A$ and node $s$ satisfy conditions (i)-(iv). By \cref{claim:subtreeOfUpwardsForest} and \cref{claim:forestProp}, $H_B^{\uparrow}$ is a subtree of a connected component of $F^{\uparrow}$ which  means that $H_{B}^\uparrow$ is a subtree of $H_{A_{s'}}^\uparrow$ for some $s'\in S$.  In particular, this implies that $B\subseteq A_{s'}$. Finally, $B$ must in fact be equal to $A_{s'}$ because otherwise $A_{s'}$ is a proper superset of $B$ which together with $s'$ satisfies properties (i)-(iv) by \cref{claim:prop(i)-(iv)}. Additionally, $s=s'$ because $s$ is the tip of $B$ while $s'$ is the tip of $A_{s'}=B$ and hence $B=A_s$. 
\end{proof}

\subsection{The colouring}
In this subsection, we describe the colouring we use for our transduction and introduce the predicate we obtain that identifies representative sets for all inner nodes. For this, we first show that we can partition the nodes of every laminar tree into a finite number of thin sets. 

\begin{lemma}\label{lem:partitionThinSets}
    Let $T$ be the laminar tree of a set system $(U,\mathcal{F})$ and $L$ the set of leaves of $T$. We can partition the set of inner nodes $V(T)\setminus L$ into $16$ (possibly empty) thin sets.
\end{lemma}

\begin{proof}
    For every $i\in \{0,1,2,3\}$, let $V^i\subseteq V(T)\setminus L$ be the set of inner nodes with $\lvlParity_T(t)=i$. We further denote the root of $T$ by $r$ and the children of $r$ by $s_1,\dots, s_\ell$. We show that for every $i\in \{0,1,2,3\}$ we can partition $V^i$ into $4$ thin sets. 
    
    Hence fix $i\in \{0,1,2,3\}$ and let $j_1\equiv i-1\mod 4$ and $j_2\equiv i-2\mod 4$.
    We first partition $V^i$ into two sets $P^\LHS,P^\RHS$ as follows. For every node $t\in V^{j_1}$ (these are the parents of nodes in $V^i$) consider the children $t_1,\dots, t_k$ of $t$ which are contained in $V^i$. Note that every child of $t$ is either a leaf of $T$ or contained in $V^i$. We add $t_1$ to $P^\LHS$ and $t_2,\dots, t_k$ to $P^\RHS$. In case $t$ has only one child in $V^i$ (all other children are leaves of $T$), no node gets added to $P^\RHS$ and in case $t$ has no children in $V^i$ (all children of $t$ are leaves of $T$), neither $P^\LHS$ nor $P^\RHS$ receive an additional node. We further add $r$ to $P^\LHS$ if $i=0$ ($r$ is the only node in $V^i$ with no parent in $V^{j_1}$). By construction, the partition $(P^\LHS,P^\RHS)$ of $V^i$ has the property that every node in $P^\LHS$ has a sibling which is not contained in $P^\LHS$ with the exception of $r$ in case $i=0$ and the same holds for $P^\RHS$. For this we crucially rely on \cref{obs:atLeastTwoChildren} and additionally observe that the sibling of some node $t\in P^\LHS$, which is not contained in $P^\LHS$, can either be a leaf or contained in $P^\RHS$ and similarly for $P^\RHS$.

    We now refine the partition $(P^\LHS,P^\RHS)$ of $V^i$ into a partition $(S^\LL,S^\LR,S^\RL,S^\RR)$ as follows. For every node $t$ in $V^{j_2}$ (these are grandparents of nodes in $V^i$) consider the children $t_1,\dots, t_k$ of $t$ that have at least one child in $V^i$ ($k$ could be $0$ if all children of $t$ are either leaves or only have leaves as children). For $t_1$ we put all children contained in $P^\LHS$ into $S^\LL$ and all children contained in $P^\RHS$ into $S^\LR$. Similarly, we put all children of $t_2,\dots, t_k$ which are contained in $P^\LHS$ into $S^\RL$ and all children of $t_2,\dots, t_k$ that are contained in $P^\RHS$ into $S^\RR$. Finally, $r$ gets added to $S^\LL$ in case $i=0$ and all children among $s_1,\dots, s_\ell$ of $s$ contained in $P^\LHS$ get  added to $S^\LL$ while all children of $s$ contained in $P^\RHS$ get added to $S^\LR$ in case $i=1$ ($s,s_1,\dots, s_\ell$ are all the nodes in $V^i$ that cannot have a grandparent in $V^{j_2}$). For an illustration of the construction, see \cref{fig:partition}. We now argue that each set of the partition is thin.
    
    Clearly $(S^\LL,S^\LR,S^\RL,S^\RR)$ is a refinement of $(P^\LHS,P^\RHS)$ and therefore still satisfies the condition that every node in $S^
    \ell$ has a sibling that is not contained in $S^\ell$, with the exception of $r\in S^\LL$ in the case $i=0$.
    Hence, $S^\LL,S^\LR,S^\RL$ and $S^\RR$ each satisfy condition (TS2). Furthermore, by construction, we also satisfy that the parent $p_t$ of every node $t\in S^\LL$ has a sibling which is not contained in $S^\LL$ with the exception of $r$ and $s_1,\dots, s_\ell$ and the same holds for $S^\LR,S^\RL$ and $S^\RR$. To see this, consider $t\in S^\LL\cup S^\LR$, the parent $p_t$ of $t$ and the grandparent $q_t$ of $t$. By \cref{obs:atLeastTwoChildren} $q_t$ must have an additional child $p\not=p_t$. This additional child $p$ is either a leaf (in which case none of $p$'s children are contained in $S^\LL$ or $S^\LR$) or by construction all non-leaf children of $p$ are contained in $S^\RL$ and $S^\RR$. Now consider $t\in S^\RL\cup S^\RR$, the parent $p_t$ of $t$ and the grandparent $q_t$ of $t$. By construction all non-leaf children of $t$ are in $S^\RL\cup S^\RR$.
    Therefore, by construction, $q_t$ must have a child $p\not=q_t$ (the "first" child in our construction) for which all non-leaf children are contained in $S^\LL\cup S^\LR$. Hence, each part $S^\LL,S^\LR,S^\RL,S^\RR$ also satisfies (TS3). Finally, (TS1) is trivially satisfied for each part $S^\LL,S^\LR,S^\RL,S^\RR$ and they are therefore thin sets.
\end{proof}
\begin{figure}
    \includegraphics[width=\textwidth]{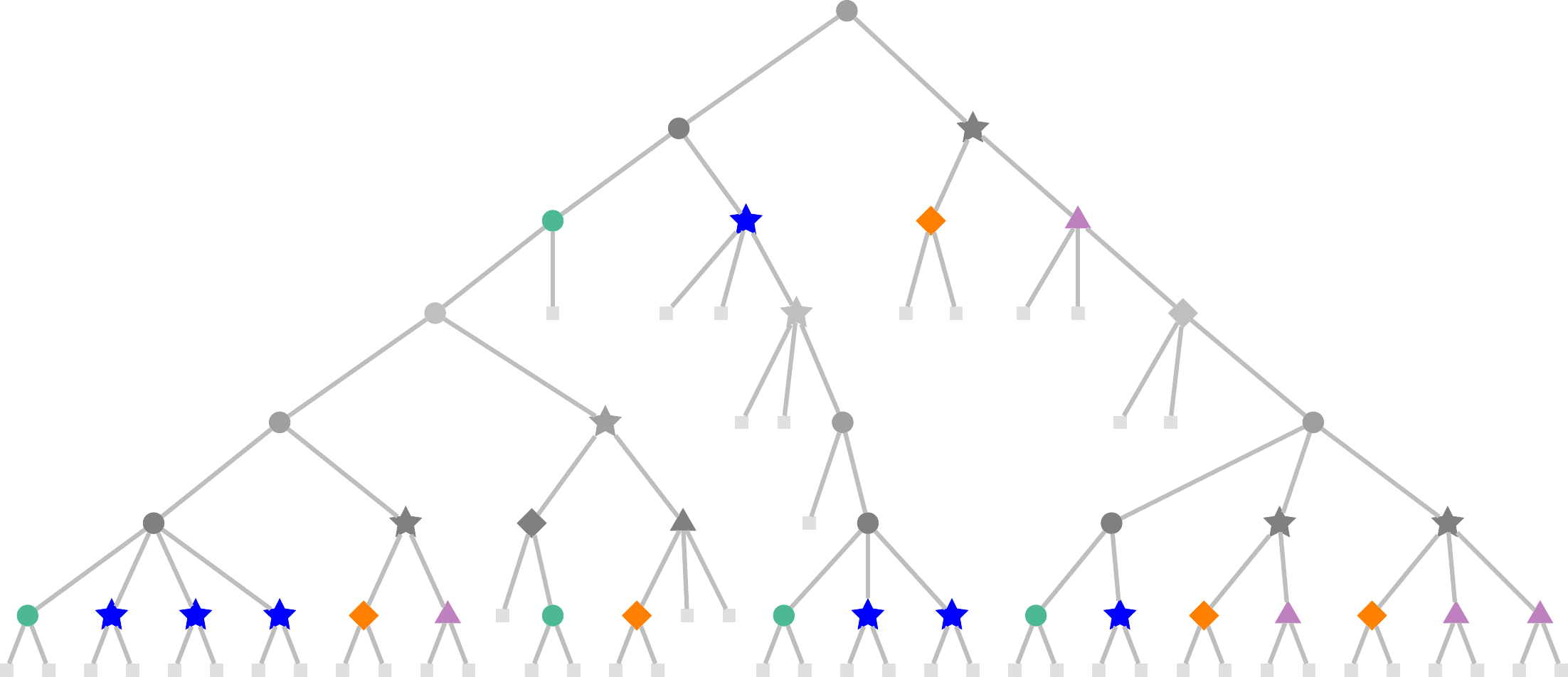}%
    \caption{%
        A partition of the inner nodes of a laminar tree into thin sets where for $i=2$ the partition $(S^\LL,S^\LR,S^\RL,S^\RR)$ of $V^i$ is highlighted in colour and the other parts of the partition into thin sets are hinted at using different shapes and shades of gray. %
    }%
    \label{fig:partition}
\end{figure}
 
Let $(U, \mathcal{F})$ be a set system and $T$ the laminar tree of $(U,\mathcal{F})$ where we denote the set of leaves of $T$ by $L$. By \cref{lem:partitionThinSets} we can partition the inner nodes $V(T)\setminus L$ of $T$ into at most $16$ (potentially empty) thin sets $S^1,\dots, S^{16}$ that correspond to the sets $S^\LL,S^\LR,S^\RL,S^\RR$ for each of the levels $V^0,V^1,V^2,V^3$. We call such a partition of $V(T)\setminus L$ a \emph{thin partition} for $(T,\mathcal{F})$. 
Our colouring of elements of $U$ consists of $32$ colours $A^1,\dots, A^{16},B^1,\dots,B^{16}$.
Here $A^i$ is the union of representative sets $(A^i_s)_{s\in S^i}$ chosen to satisfy \cref{lem:disjointRepSets} for every $i\in [16]$ and each $B^i$ contains exactly one (arbitrary) element per representative set $A^i_s$, $s\in S^i$. We call such a $32$-colouring of $U$ an \emph{identifying colouring} for the thin partition $(S^1,\dots, S^{16})$. An \emph{identifying colouring} is simply an identifying colouring obtained from any partition $(S^1,\dots, S^{16})$ of $V(T)\setminus L$ into thin sets. \\

We now show how an identifying colouring can be used to identify in \MSO one representative element of $U$ for each inner node of the laminar tree of a laminar set system $(U,\mathcal{F})$. For this, we use representative sets and then the arbitrarily chosen element from each representative set as the \emph{leader} which will play the role of the representative element. Before we proceed, let us first provide some auxiliary predicates.

For a set system $(U,\mathcal{F})$, we define $\rootNode(X)$ and $\leaf(X)$ to be a predicates with one free set variable expressing that $X$ corresponds to the root (or leaf, resp.) of the laminar tree of $(U,\mathcal{F})$ which is easy to implement in \MSO. Furthermore, for a unary predicate $A$ we additionally define a predicate with one free set variable $\leaf_A(X)$ that requires $X$ to correspond to a leaf of the laminar tree and the element contained in $X$ additionally needs to satisfy predicate $A$. 
We further introduce predicates $\ancestor(X,Y)$, $\descendant(X,Y)$, $\parent(X,Y)$ and $\child(X,Y)$ each having two free set variables $X,Y$ and which are satisfied if $X,Y\in \mathcal{F}$ and the node corresponding to $X$ in the laminar tree of $(U,\mathcal{F})$ is the ancestor, descendant, parent or child, respectively, of the node corresponding to $Y$.  It is routine to implement these four predicates in \MSO.

The following lemma implements predicates that provide the identification between sets in $\mathcal{F}$ and elements in $U$ in such a way that all elements in $A_i$, $i\in \{1,\dots, 16\}$ receive distinct representatives which are exactly the leaders of their respective representative sets. We use one leader predicate $\leader_i$ for every $i\in \{1,\dots, 16\}$ which provides the leaders for sets corresponding to nodes in $A_i$. We use an additional predicate $\leader_0$ which gives the identification between singleton sets and elements. While this is trivial, it is convenient later to avoid exceptions for singleton sets.
\begin{lemma}\label{lem:leaderFormula}
    Let $(U,\mathcal{F})$ be a laminar set system, $(S^1,\dots, S^{16})$ a thin partition for $(U,\mathcal{F})$ and $A^1,\dots, A^{16},B^1,\dots,B^{16}$ an identifying colouring for $(S^1,\dots,S^{16})$. 
    Using the unary predicates $A^1,\dots, A^{16},B^1,\dots,B^{16}$ we can define binary predicates $\leader_i(r,X)$ for every $i\in \{0,\dots,16\}$ such that for every $X\in \mathcal{F}$ there exists exactly one $i\in \{0,\dots,16\}$ and one $r\in U$,
    the leader of $X$, for which $\leader_i(r,X)$ is satisfied. Furthermore, for every $i\in \{0,\dots, 16\}$ the elements $r\in U$ for which there is a set $X\in \mathcal{F}$ such that $\leader_i(r,X)$ is satisfied are pairwise different.
\end{lemma}

\begin{proof}
    Our goal is to define the conditions of \cref{lem:uniqueRepSets} in \MSO. For this, note that instead of using nodes of the laminar tree, which we do not have access to, in the definition of the sentence $\leader_i(r,X)$, we rely on the fact that each node corresponds to a unique set in $\mathcal{F}$ and use sets of $\mathcal{F}$ instead of nodes. We \MSO define the leader in this way for every inner node of the laminar tree. For the leaves which correspond to all singleton sets in $\mathcal{F}$ we additionally define the trivial predicate $\leader_0(r,X):=\leaf(X)\land r\in X$ which assigns every singleton $\{r\}$ the leader $r$.\\
    
    To implement \cref{lem:uniqueRepSets} in \MSO we first require predicates that define when a set of $\mathcal{F}$ is fully branched, single branched or missed by the set defined by any of the predicates $A^i$. Hence, fix any unary predicate $A$ representing a subset of $U$. For any $A\in \{A^1,\dots, A^{16}\}$ we define three predicates with one free set variable $X$ each expressing that $X\in \mathcal{F}$ and the node corresponding to $X$ in the laminar tree of $(U,\mathcal{F})$ is fully branched, single branched or missed by $A$ as follows: 

\begin{align*}
    \fully_A(X):=&\SET(X)\land \lnot \leaf(X)\land \forall Y \Big(\child(Y,X)\rightarrow \exists y\big(y\in Y\land A(y)\big)\Big),\\
    \single_A(X):=&\SET(X)\land \lnot \leaf(X)\land \exists Y \Big(\child(Y,X)\land \exists y\big(y\in Y\land A(y)\big)\land \\
    &\phantom{aaaaaaaa}\forall Z \Big[\big(\child(Z,X)\land \exists z\big(z\in Z\land A(z)\big)\big)\rightarrow Z=Y\Big]\Big),\\
    \missed_A(X):=& \SET(X)\land \lnot \leaf(X)\land \forall x \big(x\in X\rightarrow \lnot A(x)\big).
\end{align*}
Since $A^1,\dots, A^{16}, B^1,\dots, B^{16}$ is an identifying colouring for $(S^1,\dots, S^{16})$ each $A^i$ is the union of representative sets $(A^i_s)_{s\in S^i}$ with pairwise disjoint up-trees $H^\uparrow_{A^i_s}$. Hence, by \cref{lem:uniqueRepSets} any set $R\subseteq U$ and set $X$ corresponding to inner node $s$ of the laminar tree satisfy $R=A_s^i$ if and only if they satisfy predicate $\repSet_A^\ast(R,X)$ defined below. Note that we first define the following auxiliary predicate $\repSet_A(R,X)$ which implements precisely conditions (i)-(iv) of \cref{lem:uniqueRepSets} and then $\repSet_A^\ast(R,X)$ extends $\repSet_A(R,X)$ to also include property (v). This is needed because property (v) states that set $R$  is  maximal among all sets for which there is a node and they satisfy (i)-(iv). We define $\repSet_A(R,X)$ as follows:
\begin{align*}
    \repSet_A(R&,X):=\SET(X)\land \\
    &\forall x(x\in R\rightarrow A(x))\land &&\textbf{\textcolor{black!90}{(i)}}\\
    &\exists Y \Big(\child(Y,X) \land R\subseteq Y\land \forall Y'\big[\big(\descendant(Y',X)\land R\subseteq Y'\big)\rightarrow Y'=Y\big] \Big)\land &&\textbf{\textcolor{black!90}{(ii)}}
    \\
    &\forall Y\Big(\big(\descendant(Y, X)\land \fully_A(Y)\big)\rightarrow \forall Z\Big[\child(Z,Y)\rightarrow \big(\leaf_A(Z)\lor \single_A(Z)\big)\Big]\Big)\land &&\textbf{\textcolor{black!90}{(iii)}}\\
    &\forall Y\Big(\big(\descendant(Y, X)\land \single_A(Y)\big)\rightarrow \exists Z\Big[\child(Z,Y)\land \big(\leaf_A(Z)\lor \fully_A(Z)\big)\land \\
    &\phantom{aaaaaaaaaaaaaaaaaaaaaaa}\forall Z'\big((\child(Z',Y)\land Z'\not= Z)\rightarrow \missed_A(Z')\big)\Big]\Big). &&\textbf{\textcolor{black!90}{(iv)}}
\end{align*}
Note that each line of the formula above expresses  precisely one condition from \cref{lem:uniqueRepSets} (excluding (v)) and we additionally have to ensure that $X$ corresponds to a node and therefore is in $\mathcal{F}$. We define $\repSet^\ast_A(R,X)$ as follow: 
\begin{align*}
    \repSet_A^\ast(R,X):=\repSet_A(R,X)\land \forall X'\forall R'\Big(\big(\ancestor(X',X)\land &\repSet_A(R',X')\big)\rightarrow \\
    &\big(R'=R\land X'=X\big)\Big).&&\textbf{\textcolor{black!90}{(v)}}
\end{align*}
For a fixed $i\in [16]$ \cref{lem:uniqueRepSets} implies that for every $X\in \mathcal{F}$ corresponding to node $s$ of the laminar tree there is a unique set $R$, namely $R=A^i_s$, for which $\repSet^\ast_{A^i}(R,X)$ is satisfied. Note that since for every $i\in [16]$ the up-trees $H^\uparrow_{A^i_s}$ are pairwise disjoint for all $s\in S^i$ the representative sets $(A^i_s)_{s\in S^i}$ must also be pairwise disjoint. 
Hence, by choosing an arbitrary element in $A^i_s$ as the leader, as the unary predicates $B^1,\dots,B^{16}$ do,  we obtain that each set $X\in \mathcal{F}$ obtains a unique element as leader (its precisely the element contained in $A^i_s$ which satisfies predicates $B^i$).
Hence, for $i\in [16]$ we can define the predicate $\leader_i(r,X)$ as follows: 
\begin{align*}
    \leader_i(r,X):=B^i(r)\land \exists R\big(r\in R\land  \repSet_{A_i}^\ast(R,X)\big).
\end{align*}
Note that since for fixed $i\in [16]$ the $A^i_s$ are pairwise disjoint, it is also guarantied that for any fixed $i\in [16]$ all leaders of nodes in $A^i$ are pairwise different. This concludes the proof of the lemma.
\end{proof}

\subsection{The transduction}
We have introduced all the tools to provide our transduction. We recall \cref{thm:main} here in a slighly different version. We note however that by a simple interpretation which forgets the $\SET$ predicate, we can obtain the original version.
\begingroup
\def\thetheorem{\ref{thm:main}}
\begin{theorem}
There is an overlay \MSO-transduction $\tau$ from $\{\SET\}$-structure to $\{\descendantT, \SET\}$-structures such that for any laminar set system $(U,\mathcal{F})$ represented by the $\{SET\}$-structure $\mathbb{F}$ the image of $\tau$ is the $\{\descendantT, \SET\}$-structure $\mathbb{F}\sqcup \mathbb{T}$ where $\mathbb{T}$ is the $\{\descendantT\}$-structure which represents the laminar tree $T$ of $(U,\mathcal{F})$.
\end{theorem}
\addtocounter{theorem}{-1}
\endgroup

\begin{proof}[Proof of \cref{thm:main}]
    Let $(U,\mathcal{F})$ be a laminar set system and $\mathbb{F}$ be the $\{\SET\}$-structure representing $(U,\mathcal{F})$.
    Our \MSO-transduction is obtained by composing the following atomic \MSO-transductions.
    To define the transduction, we use the formula constructed in \cref{lem:leaderFormula}. 
    \begin{enumerate}
        \item We first  apply a $32$-colouring transduction to structure $\mathbb{F}$ and obtain unary predicates $\colorV_1,\dots, \colorV_{32}$.
        \item We use filtering to ensure that the colouring $\colorV_1,\dots, \colorV_{32}$ is essentially identifying leaders correctly. To achieve this, we need to check that for every set $X\in \mathcal{F}$ there is precisely one $i\in \{0,\dots,16\}$ and one $r\in U$ for which the predicate $\leader_i(r,X)$ is satisfied. Furthermore, we need to ensure that for every $i\in \{0,\dots,16\}$ and any two distinct sets $X,X'\in \mathcal{F}$ for which there are elements $r,r'$ such that $\leader_i(r,X)$ and $\leader_i(r',X')$ are satisfied, it holds that $r\not=r'$. Note that we can express both these properties through an \MSO sentence $\chi$. By \cref{lem:partitionThinSets} and \cref{lem:disjointRepSets} there must always exist an identifying colouring while \cref{lem:leaderFormula} guarantees that if $\colorV_1,\dots, \colorV_{32}$ is an identifying colouring then the properties we filter by must be satisfied. Hence, for any laminar set system $\mathbb{F}$ we must receive at least one output.
        \item Copy the resulting structure $16$ times (resulting in $17$ instances of each element of $U$),
        thus introducing $16$ additional binary relations~$(\copyV_i)_{i\in \{0,\dots,16\}}$
        where $\copyV_i(x, y)$ for $i>0$ indicates that~$x$ is the $i$-th copy
        of the original element~$y$ and $\copyV_0(x, y)$ indicates that~$x$ is the original element $y$. Note that the original elements will represent the leaves of the laminar tree while the $i$-th copy of the leader of a set $X\in \mathcal{F}$ will play the role of the node $t$ corresponding to $X$ in case the appropriate thin set $S^i$ contains $t$.
        \item Finally, we interpret using a pair of formulas $(\phi,\psi_\descendant)$ to obtain a $\{\descendant\}$ structure $\mathbb{T}$. For this we define the \emph{realization} of a set $X\in \mathcal{F}$ to be the $i$-th copy of the element $r\in U$ for which $\leader_i(r,X)$ is satisfied. We define predicate $\realization(x,X)$ expressing that $x$ is the relization of $X$ as follows:
        \begin{align*}
            \realization(x,X):= \exists r\Big(\bigvee_{i\in \{0,\dots,16\}}\big(\leader_i(r,X)\land \copyV_i(x,r)\big)\Big).
        \end{align*}
        Now $\phi$ restrict the universe to elements which are the realization of some set in $\mathcal{F}$. We define the descendant relationship of $\mathbb{T}$ as follows. For two elements $x,x'$ of the restricted universe, $\psi_\descendant$ defines that $x$ is a descendant of $x'$ if they are the realization of two sets $X,X'\in \mathcal{F}$ and $X\subseteq X'$. Formally, let
        \begin{align*}
            \psi_\descendantT(x,x'):=\exists X \exists X' \Big(\realization(x,X)\land \realization(x',X') \land X\subseteq X'\Big).
        \end{align*}
        \end{enumerate}
        We are left to argue that the resulting structure $\mathbb{T}$ represents the laminar tree $T$ of $(U,\mathcal{S})$.  For this we use the properties we ensured through filtering. First, we define a partition $(P^0,\dots, P^{16})$  of $\mathcal{F}$   as follows. Set $X\in \mathcal{F}$ to be contained in part $P^i$ if and only if there is $r\in U$ for which $\leader_i(r,X)$ is satisfied. Note that this is a partition because for each $X$ there is precisely one $i\in \{0,\dots,16\}$ and one  $r\in U$ for which this is the case (ensured by filtering). Because of filtering we know that for each $i\in \{0,\dots,16\}$ the leaders of sets in $P^i$ are distinct. Our interpretation ensures that the $i$-th copy of the leader of set $X\in P^i$ is its realization and therefore plays the role of the node of $T$ corresponding to $X$.  This is ensured by the realization  of $X\in P^i$ being the descendant of the realization of $X'\in P^j$ exactly when $X\subseteq X'$.
        Finally, we note that for each set in $\mathcal{F}$ there is precisely one element of $\mathbb{T}$ which realizes it in the structure $\mathbb{T}$ and hence $\mathbb{T}$ represents the laminar tree $T$.
\end{proof}

\subsection{Beyond laminar set systems}\label{sec:beyondLaminar}

In this section we shortly outline how to adjust the proofs of \cite{CampbellGKKK25} to obtain \cref{cor:partitive}, \cref{cor:biPartitive} and \cref{cor:graphDecompositions}. 

We first note that the transductions used to obtain \cite[Theorem 1]{CampbellGKKK25} uses the \CMSO-transduction $\tau$ which given a laminar set system outputs its laminar tree \cite[Theorem 2]{CampbellGKKK25} as a black box. 
As stated in \cite{CampbellGKKK25}, the other steps of the transductions in \cite[Theorem 1]{CampbellGKKK25} are \MSO-transductions. Roughly these transductions are of the form:
\begin{enumerate}
    \item An \MSO-interpretation that adds a predicate $\SET$ to recognize sets corresponding to a side of a separations (for bipartitions we take the side not containing a fixed element).
    \item An \MSO-interpretation that adds a predicate $\SET!$ that recognizes the sets that do not cross any other.
    \item As these sets are laminar, we can apply $\tau$ to transduce tree-structure.
    \item An \MSO-interpretation that adds the information to the tree-structure that is necessary to recover the original object.
\end{enumerate}

Thus, by using our new \MSO-transduction (\cref{thm:main}) within the transductions given in \cite{CampbellGKKK25} directly yields \cref{cor:partitive}, \cref{cor:biPartitive} and \cref{cor:graphDecompositions} with the exception of the transduction outputting the bi-join decomposition of a graph.
Although this transduction does not rely on counting aside from obtaining the laminar tree, it uses the representative predicate $\rep_A(x,X)$ which is defined within the transduction producing the laminar tree in a different way to the predicate given in this paper. The predicate is required for the following reason. 

A bi-join of a graph $G$ is a bipartition $\{X,Y\}$ of $G$ for which there exists subsets $X'\subseteq X$ and $Y'\subseteq Y$ such that $X'$ is complete to $Y'$, $X\setminus X'$ is complete to $Y\setminus Y'$ and there are no further edges between $X$ and $Y$. In the bi-join decomposition, an auxiliary vertex is introduced for each of the sets $X', X\setminus X', Y'$ and $Y\setminus Y'$. However, structurally the bi-join is completely symmetric and therefore the representative of the node representing a particular bi-join in the decomposition is used to identify which of the sets $X', X\setminus X', Y',Y\setminus Y'$ is associated with which of the auxiliary vertices. To conclude, we can easily use our predicate $\leader_i(r,X)$ instead for this purpose.

\section{Simulating counting quantifiers in laminar set system}\label{sec:simulatingCounting}

In this section, we consider for which classes of laminar set systems the parity of a set is expressible over \MSO.
By Theorem~\ref{thm:main} and Lemma~\ref{lem:backwardsDirection}, we have transductions between laminar set systems and laminar trees, so by the Backwards Translation Theorem (Theorem~\ref{thm:backwards_translation}), it is equivalent to ask when we can \MSO-express the parity of a set of leaves.
Using the technique of Thatcher and Wright~\cite{ThatcherWright}, we show that if we have a rooted tree with bounded down-degree, then counting is definable over \MSO.
\begin{theorem}\label{thm:simulateCounting}
    Let $k$ be a positive integer.
    Let $\mathcal{T}_k$ be the collection of rooted trees with down-degree at most $k$.
    Then there is a unary set-predicate $\op{EVEN-LEAF}_k$, that is $\{\descendantT\}$-definable over \MSO, where for any rooted tree $T\in\mathcal{T}_k$ and any subset of the leaves $X$, we have $\op{EVEN-LEAF}_k(X)$ as true if $|X|$ is even.
\end{theorem}

We first define some useful predicates in the vocabulary of rooted trees, which encode when a node is the root, when a node is a leaf, and when a node is a child of another. Let
\[
\rootT(x):=\forall y\left(\descendantT(x,y)\rightarrow y=x\right),
\]
\[
\leafT(x):=\forall y\left(\descendantT(y,x)\rightarrow y=x\right), \text{ and}
\]
\[
\childT(x,y):= \descendantT(x,y)\wedge\forall z\left(\descendantT(x,z)\wedge\descendantT(z,y)\rightarrow z=x\vee z=y\right).
\]

\begin{proof}
    We capture the parity of $|X|$ by considering it as a sum modulo two where we group terms according to $T$ and have a set $S$ keeping track of which partial sums are one mod two. We first require a way to express the sum of a bounded number of children.

    Given a fixed positive integer $\ell$, there is an $\MSO$-expression for a node $v$ having exactly $\ell$ children in a set $S$, namely:
    \[
    \op{countkids}_\ell(v,S):= \exists s_1\dots\exists s_\ell \left(\bigwedge_{i\neq j} s_i\neq s_j\right)\wedge\forall x \left(x\in S\wedge\op{child}(x,v)\longleftrightarrow\bigvee x=s_i\right).
    \]
    So given a fixed positive integer $k$, in a rooted tree with degree at most $k$ we have an \MSO-formula for a node $v$ having an odd number of children in a set $S$ by using the previous formula for all positive odd numbers less than $k$, namely:
    \[
    \op{oddkids}_k(v,S):= \bigvee_{\substack{\ell\text{ odd}\\\ell\leq k}}\op{countkids}_\ell(v,S).
    \]

    Now we express $|X|$ being even by saying
    there is a set of nodes $S$ where: a leaf of the tree is in $S$ if and only if it is in $X$, an internal node is in $S$ if and only if an odd number of its children are in $S$, and the root is not in $S$.
    Namely:
    \begin{align*}
    \op{EVEN-LEAF}_k(X):=&
        \exists S\big(
            \forall \ell \left[\op{leaf}(\ell)\rightarrow
                \left(\ell\in S\leftrightarrow\ell\in X\right)\right]\wedge\\
            &\qquad\forall v\left[\neg\op{leaf}(\ell)\rightarrow
                \left(v\in S\leftrightarrow\op{oddkids}_k(v,S)\right)\right]\wedge\\
            &\qquad\forall r\left[\op{root}(r)\rightarrow r\notin S\right]\big).
    \end{align*}
    The first two conditions certifies that $S$ is exactly the nodes of the tree with an odd number of descendants in $X$;
    by inducting on the height of a given node and observing that a sum is odd if and only if an odd number of its terms is odd.
    Thus, since all leaves are descendants of the root, the last condition indeed certifies that $X$ is even. 
\end{proof}
We remark that the proof above can easily be adapted to show \cref{thm:simulateCounting} for counting module $i$ for other $i$ than $2$.

In contrast,
we cannot define $\op{EVEN-LEAF}$ for arbitrarily large stars
where a \emph{star} is a tree consisting of a root vertex where all of its children are leaves.
This implies that if the class of rooted trees, $\cT$, is closed under subgraphs, then $\op{EVEN-LEAF}$ is \MSO-definable if and only if there is a universal bound on the degree of trees in $\cT$.
\begin{theorem}
    Let $\cT$ be a class of rooted trees.
    If $\cT$ contains arbitrarily large stars, then we cannot $\{\descendantT\}$-define $\op{EVEN-LEAF}$ for $\cT$ over \MSO.
\end{theorem}
\begin{proof}
    It is well-established that when our vocabulary is empty, then we cannot \MSO-express that the universe has even parity (see~\cite[Proposition~7.12]{Libkin}). We show that we can reduce our theorem to this fact.

    We show that we have an overlay \MSO-transduction $\tau$ from $\{\}$-structures to $\{\descendantT\}$-structure where for $\{\}$-structure $U$, the image $\tau(U)$ is a rooted tree with leaves $U$. Namely,
    \begin{enumerate}
        \item Apply a $1$-copying transduction to get binary predicates $\copyV_0$ and $\copyV_1$.
        \item Apply a $2$-colour transduction to get unary predicates $\colour_1$ and $\colour_2$.
        \item Filter to ensure that colour 1 consists of all original elements,
        and colour 2 consists of a single element $r$ of a copied element.
        \item Apply an interpretation transduction where $\phi$ is satisfied for vertices in $\colour_1$ or $\colour_2$ and 
        \[
        \psi_{\descendantT}(x,y):=\colour_1(x)\wedge\colour_2(y). 
        \]
    \end{enumerate}
    
    Suppose for contradiction that $\op{EVEN-LEAF}$ is $\{\ancestorT\}$-expressible for $\cT$ over \MSO.
    Then the sentence
    \[
    \forall X\left(
    \forall\ell\left[\leafT(\ell)\leftrightarrow\ell\in X \right]
    \rightarrow\op{EVEN-LEAF}(X)
    \right)
    \]
    expresses that the collection of all leaves is even.
    By the Backwards Translation Theorem (Theorem~\ref{thm:backwards_translation}) applied to our transduction $\tau$ above,
    this shows that there is a sentence \MSO-expressing that the universe has even parity, contradicting Proposition~7.12 from~\cite{Libkin}.
    Thus $\op{EVEN-LEAF}$ is indeed not $\{\descendantT\}$-expressible for $\cT$ over \MSO.
\end{proof}

\bibliography{bib}

@article{chein1981partitive,
  title={Partitive hypergraphs},
  author={Chein, Michel and Habib, Michel and Maurer, Marie-Catherine},
  journal={Discrete mathematics},
  volume={37},
  number={1},
  pages={35--50},
  year={1981},
  publisher={Elsevier}
}

@TechReport{Courcelle13,
  author = 	 {Courcelle, Bruno},
  title = 	 {The atomic decomposition of strongly connected
graphs},
  institution =  {Université de Bordeaux},
  year = 	 {2013},
  note = 	 {Available at https://www.labri.fr/perso/courcell/ArticlesEnCours/AtomicDecSubmitted.pdf}
}

@article{Courcelle99,
  author       = {Bruno Courcelle},
  title        = {The Monadic Second-Order Logic of Graphs {XI:} Hierarchical Decompositions
                  of Connected Graphs},
  journal      = {Theor. Comput. Sci.},
  volume       = {224},
  number       = {1-2},
  pages        = {35--58},
  year         = {1999},
  url          = {https://doi.org/10.1016/S0304-3975(98)00306-5},
  doi          = {10.1016/S0304-3975(98)00306-5},
  bibsource    = {dblp computer science bibliography, https://dblp.org}
}

@article{Courcelle06,
  title={The monadic second-order logic of graphs {XVI}: Canonical graph decompositions},
  author={Courcelle, Bruno},
  journal={Logical Methods in Computer Science},
  volume={2},
  year={2006},
  publisher={Episciences. org}
}

@inproceedings{GanzowR08,
  author       = {Tobias Ganzow and
                  Sasha Rubin},
  editor       = {Susanne Albers and
                  Pascal Weil},
  title        = {Order-Invariant {MSO} is Stronger than Counting {MSO} in the Finite},
  booktitle    = {{STACS} 2008, 25th Annual Symposium on Theoretical Aspects of Computer
                  Science, Bordeaux, France, February 21-23, 2008, Proceedings},
  series       = {LIPIcs},
  volume       = {1},
  pages        = {313--324},
  publisher    = {Schloss Dagstuhl - Leibniz-Zentrum f{\"{u}}r Informatik, Germany},
  year         = {2008},
  doi          = {10.4230/LIPICS.STACS.2008.1353},
  bibsource    = {dblp computer science bibliography, https://dblp.org}
}

@article{BojanczykGP21,
  author       = {Bojańczyk, Mikołaj and
                  Grohe, Martin and
                  Pilipczuk, Michał},
  title        = {Definable decompositions for graphs of bounded linear cliquewidth},
  journal      = {Log. Methods Comput. Sci.},
  volume       = {17},
  number       = {1},
  year         = {2021},
  bibsource    = {dblp computer science bibliography, https://dblp.org}
}

@inproceedings{CampbellGKKO25,
  author       = {Rutger Campbell and
                  Bruno Guillon and
                  Mamadou Moustapha Kant{\'{e}} and
                  Eun Jung Kim and
                  Sang{-}il Oum},
  title        = {Recognisability Equals Definability for Finitely Representable Matroids
                  of Bounded Path-Width},
  booktitle    = {40th Annual {ACM/IEEE} Symposium on Logic in Computer Science, {LICS}
                  2025, Singapore, June 23-26, 2025},
  pages        = {678--690},
  publisher    = {{IEEE}},
  year         = {2025},
  url          = {https://doi.org/10.1109/LICS65433.2025.00057},
  doi          = {10.1109/LICS65433.2025.00057},
  bibsource    = {dblp computer science bibliography, https://dblp.org}
}

@book {Libkin,
    AUTHOR = {Libkin, Leonid},
     TITLE = {Elements of finite model theory},
    SERIES = {Texts in Theoretical Computer Science. An EATCS Series},
 PUBLISHER = {Springer-Verlag, Berlin},
      YEAR = {2004},
     PAGES = {xiv+315},
      ISBN = {3-540-21202-7},
   MRCLASS = {03-01 (03C13 68-01 68Q19)},
  MRNUMBER = {2102513},
MRREVIEWER = {Jan\ G.\ Van den Bussche},
       DOI = {10.1007/978-3-662-07003-1},
       URL = {https://doi.org/10.1007/978-3-662-07003-1},
}

@article {ThatcherWright,
    AUTHOR = {Thatcher, J. W. and Wright, J. B.},
     TITLE = {Generalized finite automata theory with an application to a
              decision problem of second-order logic},
   JOURNAL = {Math. Systems Theory},
  FJOURNAL = {Mathematical Systems Theory. An International Journal on
              Mathematical Computing Theory},
    VOLUME = {2},
      YEAR = {1968},
     PAGES = {57--81},
      ISSN = {0025-5661},
   MRCLASS = {02.88},
  MRNUMBER = {224476},
MRREVIEWER = {L.\ H.\ Landweber},
       DOI = {10.1007/BF01691346},
       URL = {https://doi.org/10.1007/BF01691346},
}

@inproceedings{BojanczykP16,
  author       = {Bojańczyk, Mikołaj and
                  Pilipczuk, Michał},
  editor       = {Martin Grohe and
                  Eric Koskinen and
                  Natarajan Shankar},
  title        = {Definability equals recognizability for graphs of bounded treewidth},
  booktitle    = {Proceedings of the 31st Annual {ACM/IEEE} Symposium on Logic in Computer
                  Science, {LICS} '16, New York, NY, USA, July 5-8, 2016},
  pages        = {407--416},
  publisher    = {{ACM}},
  year         = {2016},
  doi          = {10.1145/2933575.2934508},
  bibsource    = {dblp computer science bibliography, https://dblp.org}
}

@article{Boj23,
  author = {Bojańczyk, Mikołaj},
  title = {The category of MSO transductions},
  year = {2023},
  month = {may},
  url = {http://arxiv.org/abs/2305.18039v1},
  date = {2023-05-29T12:01:47Z},
  eprint = {2305.18039v1},
  eprintclass = {cs.LO},
  eprinttype = {arxiv},
  urldate = {2024-07-15T08:47:21.153466Z}
}

@book{CE09,
 author = {Courcelle, Bruno and Engelfriet, Joost},
 title = {Graph Structure and Monadic Second-Order Logic},
 publisher = {Cambridge University Press},
 year = {2012},
 url = {https://doi.org/10.1017%2Fcbo9780511977619},
 doi = {10.1017/cbo9780511977619}
}

@article{Strozecki11,
  author       = {Yann Strozecki},
  title        = {Monadic second-order model-checking on decomposable matroids},
  journal      = {Discret. Appl. Math.},
  volume       = {159},
  number       = {10},
  pages        = {1022--1039},
  year         = {2011},
  url          = {https://doi.org/10.1016/j.dam.2011.02.005},
  doi          = {10.1016/J.DAM.2011.02.005},
  bibsource    = {dblp computer science bibliography, https://dblp.org}
}

@article{FunkMN22,
  author       = {Daryl Funk and
                  Dillon Mayhew and
                  Mike Newman},
  title        = {Tree Automata and Pigeonhole Classes of Matroids: {I}},
  journal      = {Algorithmica},
  volume       = {84},
  number       = {7},
  pages        = {1795--1834},
  year         = {2022},
  url          = {https://doi.org/10.1007/s00453-022-00939-7},
  doi          = {10.1007/S00453-022-00939-7},
  bibsource    = {dblp computer science bibliography, https://dblp.org}
}

@article{Hlineny06,
  author       = {Petr Hlinen{\'{y}}},
  title        = {Branch-width, parse trees, and monadic second-order logic for matroids},
  journal      = {J. Comb. Theory {B}},
  volume       = {96},
  number       = {3},
  pages        = {325--351},
  year         = {2006},
  url          = {https://doi.org/10.1016/j.jctb.2005.08.005},
  doi          = {10.1016/J.JCTB.2005.08.005},
  bibsource    = {dblp computer science bibliography, https://dblp.org}
}

@Misc{CourcelleT12,
  author = 	 {Courcelle, Bruno},
  title = 	 {Canonical Graph Decompositions},
 howpublished = {Talk},
  year = 	 {2012},
  note = 	 {Available at https://www.labri.fr/perso/courcell/Conferences/ExpoCanDecsJuin2012.pdf}
}

@article{TMSOLOF5,
  title={The monadic second-order logic of graphs {V}: On closing the gap between definability and recognizability},
  author={Courcelle, Bruno},
  journal={Theoretical Computer Science},
  volume={80},
  number={2},
  pages={153--202},
  year={1991},
  publisher={Elsevier}
}

@article{TMSOLOG1,
  title={The monadic second-order logic of graphs. {I}. Recognizable sets of finite graphs},
  author={Courcelle, Bruno},
  journal={Information and computation},
  volume={85},
  number={1},
  pages={12--75},
  year={1990},
  publisher={Elsevier}
}

@inproceedings{CampbellGKKK25,
  author       = {Rutger Campbell and
                  Bruno Guillon and
                  Mamadou Moustapha Kant{\'{e}} and
                  Eun Jung Kim and
                  Noleen K{\"{o}}hler},
  editor       = {Olaf Beyersdorff and
                  Michal Pilipczuk and
                  Elaine Pimentel and
                  Kim Thang Nguyen},
  title        = {{CMSO}-Transducing Tree-Like Graph Decompositions},
  booktitle    = {42nd International Symposium on Theoretical Aspects of Computer Science,
                  {STACS} 2025, March 4-7, 2025, Jena, Germany},
  series       = {LIPIcs},
  volume       = {327},
  pages        = {22:1--22:18},
  publisher    = {Schloss Dagstuhl - Leibniz-Zentrum f{\"{u}}r Informatik},
  year         = {2025},
  url          = {https://doi.org/10.4230/LIPIcs.STACS.2025.22},
  doi          = {10.4230/LIPICS.STACS.2025.22},
  bibsource    = {dblp computer science bibliography, https://dblp.org}
}

@article{DBLP:journals/mst/CourcelleMR00,
  author       = {Bruno Courcelle and
                  Johann A. Makowsky and
                  Udi Rotics},
  title        = {Linear Time Solvable Optimization Problems on Graphs of Bounded Clique-Width},
  journal      = {Theory Comput. Syst.},
  volume       = {33},
  number       = {2},
  pages        = {125--150},
  year         = {2000},
  url          = {https://doi.org/10.1007/s002249910009},
  doi          = {10.1007/S002249910009},
  bibsource    = {dblp computer science bibliography, https://dblp.org}
}

@article{DBLP:journals/tcs/Simon90,
  author       = {Imre Simon},
  title        = {Factorization Forests of Finite Height},
  journal      = {Theor. Comput. Sci.},
  volume       = {72},
  number       = {1},
  pages        = {65--94},
  year         = {1990},
  url          = {https://doi.org/10.1016/0304-3975(90)90047-L},
  doi          = {10.1016/0304-3975(90)90047-L},
  bibsource    = {dblp computer science bibliography, https://dblp.org}
}

\end{document}